\documentclass[reqno]{amsart}
\usepackage{amsfonts}
\usepackage{amssymb}
\usepackage[dvips]{graphics}
\usepackage{epsfig}
\usepackage{color}

\usepackage{hyperref}

\usepackage{fancyvrb, wasysym}

 \newtheorem{thm}{Theorem}[section]
 \newtheorem{cor}[thm]{Corollary}
 \newtheorem{lemma}[thm]{Lemma}
 \newtheorem{prop}[thm]{Proposition}
 \theoremstyle{definition}
 
 \newtheorem{rem} {Remark}
 
 \newtheorem{assumption}{Assumption}
 \numberwithin{equation}{section}

\newcommand{\caA}{{\mathcal A}}

\newcommand{\caF}{{\mathcal F}}

\newcommand{\caH}{{\mathcal H}}

\newcommand{\caO}{{\mathcal O}}
\newcommand{\caP}{{\mathcal P}}
\newcommand{\caQ}{{\mathcal Q}}

\newcommand{\caU}{{\mathcal U}}

\newcommand{\caZ}{{\mathcal Z}}

\newcommand{\bbC}{{\mathbb C}}

\newcommand{\bbE}{{\mathbb E}}

\newcommand{\bbN}{{\mathbb N}}

\newcommand{\bbR}{{\mathbb R}}

\newcommand{\bbT}{{\mathbb T}}

\newcommand{\bbZ}{{\mathbb Z}}

\newcommand{\iu}{\mathrm{i}}

\newcommand{\str}{^{*}}
\newcommand{\ep}[1]{\mathrm{e}^{#1}}

\newcommand{\dd}{\,\mathrm{d}}

\newcommand{\Tr}{\mathrm{Tr}}

\newcommand{\fatlattice}{\bbZ_{\scriptstyle{(\caO(L^{-\infty}))}}}

\newcommand{\be}{\begin{equation}}
\newcommand{\ee}{\end{equation}}
\newcommand{\bea}{\begin{eqnarray}}
\newcommand{\eea}{\end{eqnarray}}
\newcommand{\beann}{\begin{eqnarray*}}
\newcommand{\eeann}{\end{eqnarray*}}

\newcommand{\adjoint}{\mathrm{ad}}
\newcommand{\error}{\caO(L^{-\infty})}

\newcommand{\diam}{\mathrm{diam}}
\newcommand{\dist}{\mathrm{dist}}
\newcommand{\supp}{\mathrm{supp}}
\newcommand{\ran}{\mathrm{ran}}
\newcommand{\rk}{\mathrm{rk}}

\newcommand{\T}{\Theta}
\newcommand{\vsmall}{\caO(L^{-\infty})}
\newcommand{\eqL}{\stackrel{\scriptscriptstyle L}{=}}

\newcommand{\totalgraph}{\Lambda}
\newcommand{\halftorus}{\Gamma}
\newcommand{\caatotal}{\caA}
\newcommand{\fraction}{p}

\title{Rational indices for quantum ground state sectors}

\author{Sven Bachmann}
\address{Department of Mathematics \\ The University of British Columbia \\ Vancouver, BC V6T 1Z2 \\ Canada}
\email{sbach@math.ubc.ca}

\author{Alex Bols}
\address{University of Copenhagen \\ DK-2100 Copenhagen \O \\ Denmark}
\email{alex-b@math.ku.de}

\author{Wojciech De Roeck}
\address{ Instituut Theoretische Fysica, KULeuven  \\
3001 Leuven  \\ Belgium }
\email{wojciech.deroeck@kuleuven.be}

\author{Martin Fraas}
\address{Department of Mathematics \\ Virginia Tech \\ Blacksburg, VA 24061-0123 \\ USA}
\email{fraas@vt.edu}

\date{\today}

\begin{document}

\begin{abstract}
We consider charge transport for interacting many-body systems with a gapped ground state subspace which is finitely degenerate and topologically ordered. To any locality-preserving, charge-conserving unitary that preserves the ground state space, we associate an index that is an integer multiple of $1/\fraction$, where $\fraction$ is the ground state degeneracy. We prove that the index is additive under composition of unitaries. This formalism gives rise to several applications: fractional quantum Hall conductance, a fractional Lieb-Schultz-Mattis theorem that generalizes the standard LSM to systems where the translation-invariance is broken, and the interacting generalization of the Avron-Dana-Zak relation between Hall conductance and the filling factor. 
\end{abstract}

\maketitle

\section{Introduction}\label{Sec: Intro}

The use of topology to study condensed matter systems is among the most influential developments of late 20th century theoretical physics   \cite{Haldane_Nobel, Laughlin_Nobel}. The first major application of topology appeared in the context of the quantum Hall effect \cite{TKNN, Thouless85, AvronSeilerSimon83} in the early 80', and topological concepts have since been applied systematically to discover and classify \emph{phases of  matter}
\cite{Laughlin_Fractional,read1999beyond, haldane1983fractional,wen1989vacuum,wen1989chiral,moore1991nonabelions,frohlich1997classification}.
The full classification for independent fermions is well developed, in particular by K-theory \cite{KitaevTable,schnyder2008classification,heinzner2005symmetry}, but a framework of similar scope is lacking for interacting systems, except possibly in 1 dimension where there is a classification of matrix product states \cite{chen2011classification,TurnerFermions,fidkowski2011topological}
and cellular automata \cite{gross2012index,cirac2017matrix}.
 
For non-interacting systems, several topological indices can be formulated as \emph{Fredholm-Noether indices} \cite{NoetherIndex,ASSIndex, BellissardSchuba} or, equivalently, as transport through a \emph{Thouless pump} \cite{Thouless83}. These formulations have been influential and insightful, in particular for non-translation-invariant systems \cite{ProdanBook}.  For example, the quantum Hall conductance \cite{ASS90}, the $\bbZ_2$-Kane-Mele index \cite{katsura2016,de2015spectral}, and the particle density can be expressed as (integer-valued) Fredholm indices. 
Let us briefly recall the setting in the easiest case of one-dimensional systems, playing on the Hilbert space $\ell^2(\bbZ)$. It takes as input a self-adjoint projection $P$,  that we think of as a Fermi projection of a local and gapped one-particle Hamiltonian, and a unitary $U$ that commutes with $P$, $[P,U]=0$.  The index is then defined as $\Tr[P(U^*1_{\bbN}U-1_{\bbN})]$, provided that $[U,1_{\bbN}]$ is trace-class, with $1_\bbN$ the orthogonal projection on the subspace $\ell^2(\bbN)$. The upshot is then that, on the one hand, this index is integer-valued, and, on the other hand, it can be identified with the average charge transported by $U$ across the origin, when starting from the fermionic state defined by the projection $P$, i.e.\ the filled Fermi sea; see~\cite{ASSIndex,KitaevHoneycomb,OurAnyons} for proofs and details.

In this paper we develop an interacting analogue to this formalism.  It is similar to the non-interacting theory in that it is very modular. 
Instead of a Fermi projection, it takes as input the ground state projection of a gapped many-body Hamiltonian. Instead of the one-particle unitary, we consider a many-body unitary $U$ commuting with $P$. We also need a notion of locally conserved charge which was absent in the single-particle setting, as the charge there was implicitly defined as the number of fermions. 
 The index is then constructed out of these data ($P,U$ and the charge) under appropriate regularity conditions. 
A striking difference is the possibility that the projector $P$ has a higher rank $\fraction>1$, corresponding to degenerate or quasi-degenerate ground states. 
In this situation,  the possible values of the index are now in $\tfrac1p \bbZ$. There are two generic situations that lead to a finite $p>1$, namely spontaneous breaking of a discrete symmetry and topological order. 
The first case can morally speaking be reduced to the case of $p=1$ by restricting to superselection sectors, in which case $U$ can fail to be a symmetry and only some power $U^n$ survives as a symmetry in the superselection sector. This leads to a rational index in a straightforward way. We cover this case explicitly in Section \ref{sectors}. 
The case of topological order \cite{WenTopological,Bravyi:2011ea} is more interesting. It means that the different vectors in the range of $P$ cannot be distinguished by local observables. Our framework is tailored towards this case, allowing for example for fractional quantum Hall conductance. In~\cite{HastingsMichalakis}, fractional quantization is discussed in the same setting as here but the proof sketched there relies on a different more involved strategy.
We would however like to point out that, at present, we do not know any model where we can rigorously confirm that this index takes a non-integer value in a topologically ordered subspace $\ran(P)$, as one expects to happen in fractional quantum Hall systems.

The present paper generalizes our previous work~\cite{BBDF} which was restricted to rank-$1$ projections. Just like there, different choices for $U$ correspond to different physical situations: Our index is the (fractional) quantum Hall conductance when $U$ is associated with an adiabatic increase of flux; it is the ground state filling factor when $U$ is a discrete translation, and the theorem is a new, fractional and multidimensional version of the Lieb-Schultz-Mattis theorem, see~\cite{HastingsLSM,BrunoLSM} for the integral case.
We go further and also prove that the index is additive under composition of unitaries, as it should be if it has the interpretation of a charge transport. Applied to the specific situation of a family of covariant Hamiltonians, the additivity yields a relation between two indices. In the context of the quantum Hall effect, it relates the filling factor to the quantum Hall conductance. This is well-known in the non-interacting setting \cite{AvronDanaZak}. Here, the relation is shown to hold in the interacting, possibly fractional setting, see also~\cite{Oshikawa} for similar results as well as~\cite{WatanabeDana} for a geometric perspective.

 Finally, we mention a technical difference with the non-interacting theory briefly discussed above.
 We do not formulate our theory in infinite volume from the start, mainly because the above concepts `ground state projector' and `many-body unitary' are in general not available in a Hilbert space setting if the volume, or more precisely the number of particles, is infinite. When fixing a reference state, one can consider a Hilbert space given by the GNS representation, but that is not versatile enough for our purposes, except possibly in the nondegenerate setting, $\fraction=1$, in general in one dimension and for some two-dimensional systems \cite{ogatapersonal}. Rather, we work here in large but finite volumes and all bounds are uniform in the volume. Strictly speaking, the index is therefore associated with sequences of operators rather than just the three operators $P,U$ and the charge. 
 

\section{Setting}\label{Sec: Setting}

Let ${\totalgraph}$ be a graph  equipped with the graph distance and having diameter $\diam({\totalgraph})=L$. We write $|{\totalgraph}|$ for the number of vertices. 
 With each vertex, we associate a copy of the Hilbert space $\bbC^n$. We denote by $\caH_{\totalgraph}$ the total Hilbert space of the system of dimension $n^{|{\totalgraph}|}$.
 We treat simultaneously spin systems, where $n$ is the number of components of the spin at each site, and fermionic systems, where $n = 2^f$ with $f$ the number of flavours of fermions.

The spatial structure of ${\totalgraph}$ is reflected in the algebra of observables $\caA$. To any element $O\in\caatotal$, we associate a spatial support $\supp(O) \subset {\totalgraph}$. The crucial property is the following: If $X,Y\subset{\totalgraph}$ are disjoint, and if $\supp(O_X)\subset X,\,\supp(O_Y)\subset Y$, then $[O_X,O_Y] = 0$. These notions of locality are completely standard and probably well-known to most of our readers. For convenience, a short exposition is provided in Appendix~\ref{app: locality}.

We will consider sequences of {models indexed by $L\in\bbN$} and be interested in their asymptotic properties as $L \to \infty$.  Writing the index $L$ everywhere would clutter the text and we choose not to do so. We will always have such a family in mind and the upshot of our results is that, unless otherwise specified, all constants and parameters can be chosen independently of $L$. In particular, the parameters $R_H,R_Q, m_H, m_Q, \gamma, \fraction$ to be introduced below are assumed to be independent of $L$. Furthermore, we will use constants $c,C$, whose value can change between equations, but they are also always independent of $L$. 
We will often say
$A = \caO(L^{-\infty})$ which means that the sequence of operators $A=A_L$ satisfies
$
\Vert A_L \Vert\leq \frac{C_k}{L^k}$
pointwise for all $k\in\bbN$. For completeness, we explain the large $L$ setting in more details in Appendix~\ref{app:L}.  

With these considerations in mind we now set up the main objects of our work. We keep this section abstract on purpose and refer the reader to Section~\ref{Sec:ADZ} or to~\cite{BBDF} for specific examples.

\subsection{Hamiltonian} \label{sec: hamiltonian} The Hamiltonian is a sum of local, finite range terms of the form
\begin{equation}\label{local Hamiltonian}
H:=\sum_{Z\subset{\totalgraph}} h_{Z},
\end{equation}
where
\begin{equation*}
\supp(h_Z) = Z,\qquad h_Z = 0 \text{ unless }\diam(Z) \leq R_H,
\end{equation*}
that are uniformly bounded: $\sup_{Z\subset\Lambda}\Vert h_Z\Vert \leq m_H$. Note that $R_H$ stands here for the interaction range, not the Hall conductance. As $\Lambda$ is assumed to be finite dimensional, see Section~\ref{sec: spatial} below, it follows that $\Vert H \Vert \leq C\vert {\totalgraph} \vert$.

\subsection{Charge}\label{sec:charge}
We consider local charge operators $q_Z$ with $\supp(q_Z)=Z$ satisfying
\begin{enumerate}
\item $q_Z=0$ unless $\diam(Z) \leq R_Q$,
\item $\sup_{Z\subset\Lambda}\Vert q_Z\Vert \leq m_Q$,
\item $\sigma(q_Z)\subset\bbZ$ for all $Z$, where $\sigma(A)$ denotes the spectrum of $A$,
\item $[q_Z,q_{Z'}] = 0$ for all $Z,Z'$.
\end{enumerate}
The total charge in $S\subset{\totalgraph}$ is defined as
\begin{equation}\label{local charge}
Q_S:= \sum_{Z \subset S}  q_Z.
\end{equation}
Finally, we assume that the Hamiltonian conserves this charge, namely
\begin{equation*}
[Q_\Lambda, H] = 0.
\end{equation*}
Using the properties of $q_Z$, it follows that we can choose the decomposition $H=\sum_Z h_Z$ such that
\begin{equation}\label{H charge conservation}
[Q_\Lambda,h_Z] = 0
\end{equation}
for all $Z\subset{\totalgraph}$. This implies in particular that, for any $S\subset{\totalgraph}$, the commutator $[Q_S,H]$ is supported in a strip along the boundary of $S$.

\subsection{Spatial structure}\label{sec: spatial}
For a set $S$, we define $S_{(r)}$ to be its $r$-fattening, namely
\begin{equation}\label{fattening}
S_{(r)}:=\{x\in{\totalgraph}:\dist(x,S)\leq r\}
\end{equation}
and its boundary to be
\begin{equation*}
\partial S:= S_{(1)}\cap ({\totalgraph}\setminus S)_{(1)}.
\end{equation*}
We can now state the two conditions imposed on the graph ${\totalgraph}$:
\begin{enumerate}
\item  
${\totalgraph}$ has a finite spatial dimension in the sense of $\sup_{x\in{\totalgraph}}|\{x\}_{(r)} | \leq  C (1+r)^d$ for all $r\geq0$, i.e.\ the size of balls grows at most polynomially with the radius.
\item There is a set ${\halftorus} \subset {\totalgraph}$ such that
\begin{equation}\label{Spatial structure}
\partial {\halftorus}  =\partial_{-} \cup  \partial_{+},\qquad   \dist(\partial_{-},\partial_{+}) \geq cL .
\end{equation}
\end{enumerate} 
These assumptions are illustrated in Figure~\ref{fig:GammaTorus}, in the case where ${\totalgraph}$ is a discrete $2$-torus.
\begin{figure}
  \begin{center}
    \includegraphics[width=0.35\textwidth]{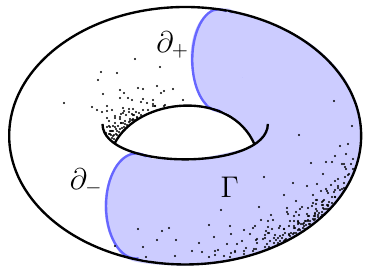}
  \end{center}
  \caption{A typical realization of the assumed global spatial structure. Here ${\totalgraph}$ is a two-dimensional discrete torus $(\bbZ/L\bbZ)^2$ and ${\halftorus}$ is one half of the torus with two disjoint boundaries $\partial_\pm$.}
  \label{fig:GammaTorus}
\end{figure}
We will consider the transport of the charge across one of ${\halftorus}$'s boundaries.

\subsubsection{Almost local operators and quasi-local unitaries}\label{sec: almost local}
We will often need to localize sequences of operators approximately, in a looser sense than by their support. To explain this, and only in this section, we keep the $L$-dependence explicit in order to be maximally clear, see also Appendix~\ref{app:L}. A sequence of operators $A = A_L$ is \emph{almost supported} in a sequence of sets $Z=Z_L$ if there are sequences $A_{r}=A_{L,r}, r\in\bbN$ with $\supp(A_{L,r}) \subset (Z_L)_{(r)}$, see~(\ref{fattening}), such that 
\begin{equation*}
\Vert A_{L}-A_{L,r}\Vert=\Vert A_L\Vert |Z_L|\caO(r^{-\infty}).
\end{equation*}
We denote\footnote{Our notation differs here from the tradition of reserving this symbol for the smaller algebra of observables that are strictly supported in $Z$.} the set of sequences of operators that are almost supported in $Z$ by $\caA_Z$. With this notation, a sequence of unitaries $U$ is called \emph{locality-preserving} if
\begin{equation*}
U\str \caA_Z U\subset\caA_Z
\end{equation*}
for all sequences of sets $Z$.  

\section{The index theorem}

With these general properties set up, we can now state the results announced above and the assumptions they require. 

\begin{assumption}[Gap]\label{ass: gap}
Let $E_1 \leq E_2 \leq \ldots \leq E_{n^{|\Lambda|}}$ be the eigenvalues of $H$, counted with multiplicities. There are $L$-independent constants  $\gamma>0,\Delta>0$ and $p$ such that
\begin{equation*}
E_{p+1}-E_{p} \geq\gamma\qquad\text{and}\qquad E_p-E_1 \leq \Delta,
\end{equation*}
and $\gamma>2\Delta$.
\end{assumption}
 We refer to the rank-$p$ spectral projector corresponding to the spectral patch $\{E_1,\ldots, E_p\}$ as $P$ and its range is called the `the ground state space'.

One consequence of the gap assumption is the following exponential clustering result. 

\begin{prop}\label{prop:clustering} 
For any $A\in\caA_X,B\in\caA_Y$ with $X\cap Y = \emptyset$ and for any normalized $\Omega \in \ran(P)$, 
\begin{equation}\label{Assumption: Clustering}
\vert\langle \Omega, AB\Omega\rangle - \langle \Omega, APB\Omega\rangle\vert = \Vert A\Vert \Vert B\Vert \vert X\vert\vert Y \vert \caO(d(X,Y)^{-\infty}).
\end{equation}
\end{prop}
This is a slightly weaker statement than~\cite{BrunoClustering,HastingsClustering} in that the decay is only superpolynomial, but also under weaker assumptions: it holds in finite volumes and with an energy width of the ground state patch $\Delta$ that may remain bounded away from $0$ in the infinite volume limit. We refer to Appendix~\ref{A:Clustering} for a proof of this finite volume clustering theorem. As seen there, the condition $\gamma>2\Delta$ is of technical nature.

The second assumption is about a locality-preserving unitary, see Section~\ref{sec: almost local}. As discussed in the introduction, this $U$ is the unitary implementing the process transporting charge, whether by translation, flux insertion, or else.
\begin{assumption}[Charge and locality preserving $U$]\label{ass: unitary}
There is a locality-preserving unitary $U$ that leaves $\ran(P)$ invariant
\begin{equation}\label{UP}
[U,P]=\vsmall
\end{equation}
and that conserves the total charge
$$
[U,Q_{{\totalgraph}}]=0.
$$
\end{assumption}
Since $U$ is locality-preserving and the charge is a sum of local terms, this leads to  
 a continuity equation:  for any spatial set $Z$, 
$$
U^* Q_{Z} U -Q_{Z}  \in \caA_{\partial Z}.
$$
In words, the net charge transported by $U$ in or out of any set is supported near the boundary of the set. Applying this assumption to $Z={\halftorus}$ and using the spatial structure introduced in Section~\ref{sec: spatial},  we get 
\begin{equation}\label{eq: splitting}
U^* Q_{{\halftorus}} U-Q_{{\halftorus}} =:T_-+T_+ 
\end{equation}
with $T_{\pm} \in \caA_{\partial_{\pm}}$ and $\Vert T_{\pm}\Vert\leq C\vert{\totalgraph}\vert$.
We naturally interpret $T_{\pm}$ as the operators of charge transport across the boundaries $\partial_\pm$.  This defines $T_\pm$ only up to vanishing tails and an arbitrary additive constant. We fix a choice such that
\begin{equation}\label{choice of T-}
\ep{2\pi\iu(Q+T_\pm)} = 1+ \vsmall,\qquad\Vert T_{\pm}\Vert\leq C\vert{\totalgraph}\vert,
\end{equation}
where we denoted $Q=Q_{{\halftorus}}$ as we shall do from here onwards. Such a choice exists. Indeed, for any $\widetilde T_\pm$ satisfying~(\ref{eq: splitting}), we have, by integrality of the spectrum of $Q$ and the assumption (\ref{Spatial structure}) about the spatial structure, that
\begin{align}
1 &= \ep{2\pi\iu U\str Q U} = \ep{2\pi\iu(Q+\widetilde T_- + \widetilde T_+ )} + \vsmall \nonumber \\
&= \ep{2\pi\iu(Q+ \widetilde T_- )}\ep{2\pi\iu(Q+\widetilde  T_+ )} + \vsmall \label{eq:nu}
\end{align}
with $\exp(2\pi\iu(Q+\widetilde T_\pm )) \in \caA_{\partial_{\pm}}$. Hence there exists $\nu$ such that 
$$
\exp(2\pi\iu(Q+\widetilde T_\pm )) = e^{\pm \iu \nu} + \vsmall.
$$
Then $T_\pm :=\widetilde T_\pm \mp \frac{\nu}{2\pi}$ satisfies both~(\ref{eq: splitting}) and~(\ref{choice of T-}). We can now state our main result. For any $\epsilon > 0$, we denote
\begin{equation*}
\bbZ_{\scriptstyle{(\epsilon)}} := \{x\in\bbR:\mathrm{dist}(x,\bbZ) < \epsilon\}.
\end{equation*}
\begin{thm}\label{thm:Index}
If Assumptions \ref{ass: gap} and \ref{ass: unitary} hold, then 
\begin{equation*}
\Tr(P T_-) \in \fatlattice.
\end{equation*}
\end{thm}
While the trace is an integer, the physically relevant quantity is the expectation value of charge transport in the state given by the density matrix $p^{-1}P$, which makes the above into a rational index indeed. We denote it by
$$\mathrm{Ind}_P(U):=\langle T(U)_- \rangle_{P},$$
where $\langle A \rangle_P := \fraction^{-1}\Tr(PA)$, further emphasizing the mathematical fact that it is a general index associated with the pair $(U,P)$ of a locality-preserving unitary and a finite-dimensional projection that commutes with the unitary.

There are two natural settings where $p=\mathrm{rk}(P)>1$: topologically ordered ground states and spontaneous symmetry breaking with a local order parameter. In both cases, a value of the index can be attributed to the individual ground states themselves. We cover here the case of topological order and we postpone symmetry breaking to Section \ref{sec:sectors}.
 \begin{assumption}[Topological order] \label{ass: toporder}
For any $Z$ such that $\diam(Z)<C$, and for any operator $A\in\caA_Z$ of norm $1$,
 \begin{equation*}
PAP - \langle A \rangle_P P = \vsmall.
\end{equation*}
\end{assumption}
This assumption prevents local order since the restriction of any local observable to the ground state space is trivial. No local observable can be used to distinguish between the different states in the range of $P$.
We note that this assumption implies that the splitting $\Delta$ in Assumption~\ref{ass: gap} vanishes in the infinite volume limit because $\langle \Omega, H \Omega \rangle= \langle H \rangle_P+\error$ by topological order for any normalized $\Omega \in \ran(P)$. 
\begin{cor}\label{cor: top order}
If, in addition to Assumptions \ref{ass: gap} and \ref{ass: unitary}, Assumption~\ref{ass: toporder} also holds, then
\begin{equation*}
\fraction \langle \Psi,  T_- \Psi \rangle \in \fatlattice
\end{equation*}
for any normalized $\Psi \in \ran(P)$. 
\end{cor}
\begin{proof} By (\ref{eq: splitting}) and the locality-preserving property of $U$, $T_-$ can be approximated by sums of local terms, to which Assumption \ref{ass: toporder} applies.
\end{proof}

We now proceed with the proof of the theorem, postponing discussions of further properties of the index and of applications.

\section{Proof of Theorem~\ref{thm:Index}}\label{Sec: Main}
 From now on, we denote by $\eqL$ equality up  to terms of $\vsmall$, in operator norm. 
\subsection{Preliminaries}\label{subsec: prelim to proof}
 Recalling that $Q$ is the charge in the half space, let us define 
\begin{equation}
\label{HastingsGenerator}
K:= \int W(t)   \ep{\iu t H} \iu[H,Q]  \ep{-\iu t H} \,dt
\end{equation}
with $W$ a real-valued, bounded, integrable function satisfying $W(t)=\caO(|t|^{-\infty})$ and $\widehat{W}(\omega)=-\frac{1}{\iu\omega}$ for all $|\omega|\geq \gamma$, with $\gamma$ the spectral gap as in Assumption~\ref{ass: gap}. Since, by functional calculus for $\mathrm{ad}_H = [H,\cdot]$,
\begin{equation*}
K = \widehat W(-\mathrm{ad}_H)(\iu\,\mathrm{ad}_H(Q)),
\end{equation*}
the properties of $W$ yield that $[K,P]=[Q,P]$.  By charge conservation~(\ref{H charge conservation}) and the spatial structure~(\ref{Spatial structure}), we see that
\begin{equation*}
\iu [H,Q]=J_-+J_+, \qquad J_{\pm}\in\caA_{\partial_{\pm}}.
\end{equation*}
Plugging this decomposition into~(\ref{HastingsGenerator}) and by the Lieb-Robinson bound, we conclude that there are $K_{\pm} \in \caA_{\partial_\pm}$ such that $\Vert K_\pm\Vert < C\vert{\totalgraph}\vert$ and such that
\begin{equation}\label{Qbar}
\overline Q : =Q-K_--K_+
\end{equation}
leaves $\ran(P)$ invariant:
\begin{equation}\label{QbarP}
[\overline Q,P]=0.
\end{equation}
Note that while~(\ref{HastingsGenerator}) is also the generator of the `quasi-adiabatic' flow, see~\cite{HastingsWen, Sven}, its use in the present context was introduced in~\cite{BBDF}.

We now present three lemmas before heading to the main argument. While the first and third ones are general and purely technical, the second one refers explicitly to the spatial structure of the problem, see Section~\ref{sec: spatial}, and plays an essential role in the following.

We note that Lemma~\ref{lemma: PUP vs commutator} is a slight modification from the published version of the present article, which is incorrect there. In consequence, there is a small change also in the proof of Lemma~\ref{lemma: V- Invariance}, the statement of the lemma remains unchanged. We thank Filippo Santi for pointing out the error.

\begin{lemma}\label{lemma: PUP vs commutator}
Let $V$ be a unitary and $P$ an orthogonal projection. Then
\begin{enumerate}
\item $\Vert [V,P]\Vert <\epsilon$ implies $\Vert PV\psi\Vert^2 >1-\epsilon$ and $\Vert PV\str\psi\Vert^2 >1-\epsilon$ for all normalized vectors such that $P\psi = \psi$.
\item $\Vert PV\psi\Vert^2 >(1-\epsilon^2) \Vert P \psi \Vert^2$ and $\Vert PV\str\psi\Vert >(1-\epsilon^2) \Vert P \psi \Vert^2$ implies that $\Vert [V,P]\Vert < 2\epsilon$.
\end{enumerate}
\end{lemma}
\begin{proof}
(i) Let $\psi = P\psi,\Vert \psi \Vert = 1$. Then,
\begin{equation*}
\Vert PV\psi\Vert^2 = \langle \psi,P\psi\rangle + \langle \psi,V\str [P,V]\psi\rangle > 1-\epsilon.
\end{equation*}
The inequality for $PV\str\psi$ follows as above with $V\str [P,V]$ replaced by $[V,P]V\str$.

(ii) For an arbitrary $\psi$,
 \begin{align*} 
 \Vert P \psi \Vert^2 = \Vert V P \psi \Vert^2 &= \Vert (1-P) V P \psi \Vert^2 + \Vert P V P \psi \Vert^2 \\
 											    & > \Vert (1-P) V P \psi \Vert^2 + (1- \epsilon^2) \Vert P \psi \Vert^2.
 \end{align*}
 This implies $\|(1-P) V P \| < \epsilon$. Exchanging $V \leftrightarrow V\str $, we also get $\|(1-P) V\str P \| < \epsilon$, and hence $\| P V (1-P) \| < \epsilon$. The statement then follows from $[V,P] = (1-P)VP - PV(1-P)$.
\end{proof}
\begin{lemma}\label{lemma: V- Invariance}
Let $V_\pm\in\caA_{\partial_\pm}$ be unitary operators and let $V := V_-V_+$. If Assumption~\ref{ass: gap} holds, then $[V,P] \eqL 0$ implies $[V_\pm,P] \eqL 0$.
\end{lemma}
\begin{proof}
By Proposition~\ref{prop:clustering}, using the notation of the proposition,
\begin{equation}\label{clustering}
\Vert PAB\Omega - PAPB\Omega\Vert = \caO(d(X,Y)^{-\infty}),
\end{equation}
holds for any normalized vector $\Omega$ in the range of $P$.

Now, the assumption $[V,P] \eqL 0$ and Lemma~\ref{lemma: PUP vs commutator}(i) imply that $\Vert PV_+V_-\Omega\Vert^2\eqL 1$ for any normalized ground state vector $\Omega$, and we conclude by clustering~(\ref{clustering}) that
\begin{equation*}
1\geq \Vert P V_-\Omega\Vert
\geq \Vert P V_+P \Vert\Vert P V_-\Omega\Vert\geq \Vert PV_+P V_-\Omega\Vert \eqL 1.
\end{equation*}
Hence $\Vert P V_-\Omega\Vert\eqL 1$. The same argument but starting with the commutator $[P,V\str]$ yields the bound $\Vert P V_-\str\Omega\Vert \eqL 1$. Therefore, $\Vert [V_-,P]\Vert \eqL 0$ by Lemma~\ref{lemma: PUP vs commutator}(ii).
\end{proof}
\begin{lemma}\label{lemma: exp of op}
Let $A = A\str$ and $P$ be an orthogonal projection. If $\Vert [A,P]\Vert<\epsilon$, then 
\begin{equation*}
\Vert [\ep{-\iu\phi A},P]\Vert\leq \phi \epsilon.
\end{equation*}
\begin{proof}
Follows immediately from the identity
\begin{equation*}
[\ep{-\iu\phi A},P]  = -\iu\int_0^\phi \ep{-\iu s A}[A,P] \ep{-\iu (\phi-s)A}\dd s
\end{equation*}
and the unitarity of the exponentials.
\end{proof}
\end{lemma}

\subsection{The main argument}\label{subsec: main}

We recall that the spatial setup of Section~\ref{sec: spatial}, see also Figure~\ref{fig:GammaTorus}, is such that the boundaries $\partial_{\pm}$ are separated by $cL$. Let now ${\totalgraph}_\pm=(\partial_{\pm})_{(cL)} \cap {\halftorus}$ be regions such that 
$\dist({\totalgraph}_{+},{\totalgraph}_-)\geq cL$ (recall that $c$ can change value from equation to equation). 
We denote
$$
Q_{\pm} :=Q_{{\totalgraph}_\pm},\qquad  Q_m := Q -Q_- -Q_+.
$$

For $\phi\in[0,2\pi]$, we set
\begin{equation}
Z(\phi) := U\str \ep{\iu\phi\overline Q}U\ep{-\iu\phi\overline Q} = \ep{\iu\phi \overline{Q}^U}\ep{-\iu\phi \overline{Q}},
\end{equation}
as well as
\begin{equation}\label{Def Zpm}
Z_\pm(\phi) := \ep{\iu\phi  {\overline Q^U_{\pm}} }\ep{-\iu\phi \overline Q_{\pm}}.
\end{equation}
Here we defined $\overline Q^U:=U\str \overline Q U$, $K_{\pm}^U:=U\str K_\pm U$, and
\begin{equation}\label{Qbars}
\overline Q_{\pm} := Q_\pm - K_\pm,\qquad
\overline Q^U_{\pm} := Q_\pm +T_\pm - K_\pm^U.
\end{equation}
To avoid later confusion, we point out that $\overline Q^U_{\pm} \neq U\str Q_- U$. With these definitions, the following identities hold: 
\begin{equation}\label{Splitting}
\overline Q = \overline Q_{-}+ Q_{m} + \overline Q_{+},\qquad  \overline  Q^U \eqL \overline Q^U_{-}+ Q_m+ \overline Q^U_{+}.
\end{equation}
The crucial point of these definitions is the commutation property 
\begin{equation}\label{commutations pm}
[Q_m,A_{\pm}]\eqL 0,\qquad [A_-,B_+]\eqL 0,
\end{equation}
with $A_{\pm}, B_\pm$ any of the above objects carrying the subscript $\pm$. This is immediate for operators in $\caA_{\partial_\pm}$ but it also holds for $Q_{\pm},\overline{Q}_\pm, \overline{Q}^U_\pm \in\caA_{{\totalgraph}_\pm}$, and hence for their exponentials by Lemma \ref{lemma: exp of op}, because these operators reduce to the charge away from $\partial_\pm$.       
This immediately leads to 
$$
Z(\phi) \eqL Z_-(\phi)Z_+(\phi).
$$

By~Assumption~\ref{ass: unitary} and by construction of $\overline Q$, see~(\ref{QbarP}), all four factors of~$Z(\phi) = U\str \ep{\iu\phi\overline Q}U\ep{-\iu\phi\overline Q}$ commute with $P$, so that 
\begin{equation*}
[P,Z(\phi)]\eqL 0.
\end{equation*}
This and Lemma~\ref{lemma: V- Invariance} now yield the first essential observation:
\begin{equation}\label{Z-P}
[P,Z_{\pm}(\phi)] \eqL 0.
\end{equation}
The second one follows by recalling the integrality of the spectrum of $Q_m$, which implies that
\begin{equation}\label{exp 2pi i Q-}
[P, \ep{2\pi\iu \overline{Q}_\pm }] \eqL 0,
\end{equation}
again by Lemma~\ref{lemma: V- Invariance} applied to $\ep{2\pi\iu\overline Q} = \ep{2\pi\iu\overline Q_-}\ep{2\pi\iu\overline Q_+}$, see~(\ref{Splitting}).

We now consider the function $\phi \mapsto \caZ_{-}(\phi):= PZ_-(\phi)P $ and let 
$D_- := \overline{Q}^U_{-}-\overline{Q}_{-}$. Then,
\begin{align*}
-\iu \frac{d}{d\phi}  \caZ_{-}(\phi) 
& = P \ep{\iu\phi  {\overline Q^U_{-}} } D_- \ep{-\iu\phi \overline Q_{-}}P \\
& =      PZ_-(\phi)\ep{\iu\phi \overline Q_{-}}   D_- \ep{-\iu\phi \overline Q_{-}} P  \\
& \eqL  \caZ_{-}(\phi) P \ep{\iu\phi \overline Q_{-}}   D_- \ep{-\iu\phi \overline Q_{-}} P  \\
& \eqL   \caZ_{-}(\phi) P \ep{\iu\phi \overline Q}   D_- \ep{-\iu\phi \overline Q} P \\
& \eqL  \caZ_{-}(\phi) \ep{\iu\phi \overline Q} P     D_- P \ep{-\iu\phi \overline Q}.
\end{align*}
The first two equalities are immediate calculations, the third one follows from~(\ref{Z-P}), the fourth one uses the commutations~(\ref{commutations pm}), the fifth one is by property~(\ref{QbarP}) of~$\overline Q$. The unique solution of this differential equation with $\caZ_-(0)=1$  is
\begin{equation}\label{veryfirst}
\caZ_-(\phi) \eqL \ep{\iu\phi(P     D_- P + \overline Q)}\ep{-\iu\phi\overline Q}.
\end{equation}
Note that both unitary factors on the right independently commute with $P$. 
It remains to study $\caZ_{-}(2\pi)$ to conclude. By the integrality of the spectrum of $Q_m,Q_+$,
\begin{equation*}
U\str \ep{2\pi\iu \overline{Q}_{-}}U 
 =  U\str \ep{2\pi\iu (\overline{Q}_{-} +Q_m+Q_+)}U
 = \ep{2\pi\iu (U\str QU - K^U_-)},
\end{equation*}
where we used the definition~(\ref{Qbars}) of~$\overline Q_-$. Since $U\str Q U \eqL Q + T_- + T_+$, the integrality of the spectrum of charge and the commutation property~(\ref{commutations pm}), we conclude that 
\begin{align}\label{at2pi}
U\str \ep{2\pi\iu \overline{Q}_{-}}U 
 & \eqL  \ep{2\pi\iu(Q_++T_+)} \ep{2\pi\iu(Q_-+T_- - K^U_-)}  \\
 & \eqL \ep{2\pi\iu(Q_-+T_- - K^U_-)},\nonumber
\end{align}
where the second equality uses the choice~(\ref{choice of T-}). Since the exponent is precisely 
$\overline{Q}_{-}^U$, see again~(\ref{Qbars}), 
we multiply from the right by~$\ep{-2\pi \iu \overline{Q}_{-}}$ to obtain
\begin{equation}
\label{eq: res as product}
  \caZ_-(2\pi) \eqL PU\str \ep{2\pi\iu \overline{Q}_{-}}U \ep{-2\pi \iu \overline{Q}_{-}}   P . 
  \end{equation}
All four unitaries on the right hand side commute with $P$, see~(\ref{exp 2pi i Q-}). Moreover, since $PV\str PVP - P = PV\str [P,V]P$ for any unitary $V$, the condition $[P,V]\eqL0$ implies that $PVP$ is invertible on $\ran(P)$ with $(PVP)^{-1} \eqL PV\str P$. By continuity of the determinant, we conclude that
\begin{equation*}
\mathrm{det}_{P}(\caZ_-(2\pi)) \eqL 1
\end{equation*}
where $\mathrm{det}_{P}(A) := \det (PAP +(1-P))$. On the other hand, (\ref{veryfirst}) and the  relation
$\mathrm{det}_{P}(\ep{A})=\ep{\Tr(PA)}$ yield
\begin{equation*}
\mathrm{det}_{P}(\caZ_-(2\pi)) \eqL \ep{2\pi\iu \Tr (PD_-P + \overline Q)}\ep{-2\pi\iu \Tr (\overline Q)} = \ep{2\pi\iu \Tr (PD_-)}.
\end{equation*}
We further observe that 
\begin{equation*}
D_- = T_- - U\str K_- U + K_-,
\end{equation*}
see~(\ref{Qbars}), so that $\Tr(PD_-)\eqL \Tr(PT_-)$ by the unitary invariance of the trace. Hence $\Tr(PT_-) \in \fatlattice$, which concludes the proof of Theorem \ref{thm:Index}.  
\begin{rem}\label{rem: how to check convention}
Had we not enforced the choice~(\ref{choice of T-}) of $T_\pm$, (\ref{eq: res as product}) would read
\begin{equation}\label{Z- at 2pi without}
\caZ_-(2\pi) \eqL  P\ep{-2\pi \iu  (Q_+ + T_+)} U^*\ep{2\pi \iu  \overline Q_-}U \ep{-2\pi \iu  \overline Q_-}P,
\end{equation}
see~(\ref{at2pi}). From (\ref{eq:nu}) 
we know that $\ep{-2\pi\iu (Q_+ + T_+)} \eqL \ep{2\pi\iu (Q_- + T_-)}$ and that they are multiples of the identity. Taking the determinant of~(\ref{Z- at 2pi without}), we conclude that
\begin{equation}\label{eq:nu_not_zero}
\ep{2\pi \iu \fraction \langle T_-\rangle_P} \eqL \ep{2\pi\iu \fraction(Q_- + T_-)}.
\end{equation}
This means that, to check that $T_\pm$ satisfies (\ref{eq: res as product}) it suffices to verify that $\fraction \langle T_-\rangle_P \in \fatlattice$.  We will use this in the proof of additivity below.
\end{rem}

\section{Additivity, filling, and Avron-Dana-Zak relations} \label{Sec:ADZ}

\subsection{Additivity of the index} If $U_1,U_2$ both satisfy Assumption~\ref{ass: unitary}, then so does their product, by Leibniz' rule for the commutators and preservation of locality.
 It is then natural to expect that the charge transported by the composed action of $U_1,U_2$ is equal to the sum of the charges transported by the action of each of them individually, see~\cite{ASSIndex} for the non-interacting case. This is indeed true if we make the choice
\begin{equation}\label{restriction of product}
T(U_2 U_1)_- :=U_2^*T(U_1)_{-} U_2 + T(U_2)_{-}.
\end{equation}
This is the content of the following proposition.

\begin{prop} \label{Prop:Additivity}
Suppose that $U_1,U_2$ both satisfy Assumption~\ref{ass: unitary}. Let $T(U_2 U_1)_-$ be defined by~(\ref{restriction of product}). Then~(\ref{eq: splitting}) and (\ref{choice of T-}) hold for $T(U_2 U_1)_-$ and 
\begin{equation}\label{2-additivity}
\mathrm{Ind}_P(U_2 U_1) \eqL \mathrm{Ind}_P(U_2) + \mathrm{Ind}_P(U_1).
\end{equation}
\end{prop}

\begin{proof}
The definitions of or $T(U_1)_\pm$ and $T(U_2)_\pm$ yield
\begin{equation*}
U_2\str U_1\str Q U_1 U_2 - Q = U_2^* (T(U_1)_- + T(U_1)_+) U_2 + T(U_2)_- + T(U_2)_+.
\end{equation*}
Since $U_2$ is locality-preserving, the choice~(\ref{restriction of product}) indeed satisfies~(\ref{eq: splitting}).
Using $[U_2,P] \eqL 0$, we get
\begin{equation}\label{Tr-additivity}
\Tr (P T(U_2 U_1)_- ) \eqL \Tr (P T(U_1)_{-} ) + \Tr (P T(U_2)_{-} ) \in \fatlattice.
\end{equation}
By Remark \ref{rem: how to check convention}, this shows that $T(U_2U_1)_-$ indeed satisfies~(\ref{choice of T-}). The additivity \ref{2-additivity} is then precisely (\ref{Tr-additivity}).
\end{proof}

\subsection{The fractional Lieb-Schultz-Mattis theorem}\label{sec: flsm}
Here, we constrain the choice of graph $\totalgraph$ to have a good notion of translation. Let ${\totalgraph}'$ be a $d-1$ dimensional graph in the sense of Section~\ref{sec: spatial}. Let ${\mathbb{T}} := \bbZ/L_1\bbZ$ be a discrete circle. Then ${\totalgraph}$ is of the form
\begin{equation*}
{\totalgraph} := {\mathbb{T}}\,\Box\, {\totalgraph}'
\end{equation*}
the cartesian product of these graphs.
We write $(x_1,x') \in \totalgraph$ with $x_1 \in {\mathbb{T}}, x'\in\Lambda'$ and we let  $\T$ be a unitary shift along ${\mathbb{T}}$.  Finally, let 
$${\halftorus}= \{ (x_1,x'):  0\leq x_1 < L_1/2, x' \in \Lambda'\}.$$
This choice is consistent with the setup of Section \ref{sec: spatial} provided we let $L_1 \geq cL$, with $L$ the diameter of $\totalgraph$.  We assume that the Hamiltonian is translation invariant, i.e.\ 
$[H,\T]=0$.  Then so is its ground state space, namely $[\T,P] = 0$, and since translation clearly preserves locality, Assumption~\ref{ass: unitary} holds for $\T$.  Moreover
\begin{equation}\label{translation}
\T\str Q \T -Q = - Q_{[0]} + Q_{[\lceil L_1/2\rceil]},
\end{equation}
where $[x_1] = \{ (x_1, x'), x' \in \Lambda'\}$ and $\lceil L_1/2\rceil$ is the smallest integer not smaller than $L_1/2$. We make the natural choice $T_- = - Q_{[0]}$ for which (\ref{choice of T-}) holds. By translation invariance, the total charge per slab ${\Lambda}'$ is
\begin{equation*}
\langle Q_{[0]}\rangle_P  = \tfrac{1}{L_1} \langle Q_{\Lambda} \rangle_P.
\end{equation*}
We then obtain the following fractional Lieb-Schultz-Mattis (LSM) theorem.
\begin{cor}
\label{cor:LSM}
If Assumptions \ref{ass: gap} holds, then 
\begin{equation*}
\fraction \langle Q_{[0]}\rangle_P \in \fatlattice.
\end{equation*}
\end{cor}
Note that $\langle Q_{[0]}\rangle_P$ in the statement cannot be expected to be convergent as the volume $|\totalgraph|$ grows, which renders the claim somewhat unfamiliar. This corollary becomes in particular useful when one has full translation invariance, i.e.\ not only in the $x_1$-direction. To be very specific, we consider two-dimensional $L_1\times L_2$-tori $\Lambda$, that is, we specify now $\Lambda' =\bbZ/L_2 \bbZ$ with diameter $L=\lfloor\frac{L_1+L_2}{2}\rfloor$. We specify a particular sequence of tori by picking a function $L\mapsto (L_1,L_2)$ satisfying $L=\lfloor\frac{L_1+L_2}{2}\rfloor$. We choose this function such that $L_2$ runs through all positive integers and that both $L_1,L_2>cL$.  
Let us now assume that the total charge density $\rho_{L}:= \frac{1}{|\Lambda|} \langle Q_\Lambda \rangle_P$ converges to a limiting density $\rho$ fast enough, namely $L (\rho - \rho_{L}) \to 0$. This is a natural assumption because for gapped systems we would expect to be able to choose the boundary conditions so that local observables in the bulk approach their thermodynamic values (almost) exponentially fast, see~\cite{Sven,de2015local}.  
With this, Corollary \ref{cor:LSM} implies
\begin{equation} \label{eq: lsm physical}
\fraction\rho \in \bbZ.
\end{equation}
Indeed, from the index theorem we get that $pL_2\rho_L\in \fatlattice$. Writing $L_2\rho_L= L_2\rho+L_2(\rho_L-\rho)$ and using that the assumption of fast convergence and $L_2>cL$, we get $(pL_2\rho)\, \mathrm{mod}\, 1= o(1)$. If $L_2$ runs through $\bbN$,  rational non-integer $p\rho$ is ruled out directly and irrational $\rho$ is ruled out by ergodicity of irrational rotations on the torus. \qed

\subsection{On the absence of topological order in one dimension}

It is widely accepted and proved in some specific settings~\cite{schuch2011classifying,chen2011classification,ogata2017class} that there is no intrinsic\footnote{as opposed to `symmetry protected' topological order which is not considered here.} topological order in one dimension. The present paper proves a particular version of this statement.  Namely, in one dimension and with topological order as in Assumption~\ref{ass: toporder}, the index $\mathrm{Ind}_P(U)$ is integer-valued, even when $p>1$.  For example, in the case of the Lieb-Schultz-Mattis theorem just discussed, this means that there is no topological charge fractionalization\footnote{The adjective `topological' cannot be omitted here. In the case of spontaneous symmetry breaking with a local order parameter, the index can be rational, see Section \ref{sectors}} in the ground state sector.

Indeed, in one dimension, the region $\partial_-$ is finite so that $P\ep{2\pi\iu\overline{Q}_{-}}P\eqL zP$ for some $z\in\bbC,\vert z\vert=1$, by Assumption~\ref{ass: toporder} (note that $[\ep{2\pi\iu\overline{Q}_{-}},P]\eqL 0$ in all dimensions, but the exponential acts in general non-trivially on $\mathrm{ran}(P)$). But then 
\begin{equation}
\label{eq: res as product again}
\caZ_-(2\pi) \eqL PU\str \ep{2\pi\iu \overline{Q}_{-}}U \ep{-2\pi \iu \overline{Q}_{-}}   P
\eqL PU\str P UP \vert z\vert^2 
\eqL P
  \end{equation}
since $U$ commutes with $P$ by Assumption~\ref{ass: unitary}. With topological order, ~(\ref{veryfirst}) reads
\begin{equation*}
\caZ_-(\phi) \eqL \ep{\iu\phi PD_-P} \eqL \ep{\iu\phi PT_-P} \eqL \ep{\iu\frac{\phi}{p} \Tr(P T_-)}P,
\end{equation*}
which, with~(\ref{eq: res as product again}), yields
\begin{equation*}
\mathrm{Ind}_P(U) \in \fatlattice
\end{equation*}
as claimed.

\subsection{Magnetic systems} 

We now discuss the case of systems with magnetic fields. Our main result is an interacting version of the Avron-Dana-Zak (ADZ) relation~\cite{AvronDanaZak} between the Hall conductance and the filling factor. We start with an extensive introduction of the setup and a presentation of the relation, leaving the general rigorous result for Section~\ref{sec: general framework}. See also \cite{Oshikawa} and~\cite{WatanabeDana} for another view on the same results.

\subsubsection{Harper/Hubbard Model}\label{sec: harper model}
We consider again $\Lambda=(\bbZ/L_1 \bbZ)\,\Box\, (\bbZ/L_2 \bbZ)$, i.e.\ the $L_1 \times L_2$ discrete torus with coordinates $1\leq x_{1,2} \leq L_{1,2}$ and unit vectors $\hat e_1=(1,0), \hat e_2=(0,1)$. We describe spinless fermions in a uniform magnetic field.   Let
\begin{equation} \label{frac phi}
\phi = 2\pi\frac{m}{n}
\end{equation}
with $m,n$ coprime integers, be the magnetic flux \emph{piercing} through the unit cell and let $L_2\Phi$ be the magnetic flux \emph{threaded} through the torus, see Figure~\ref{fig: harperhubbard}.
Then, in the Landau gauge, the Hamiltonian is
\begin{equation}
\label{eq: hubbard}
H_\Phi = t \sum_{x \in \Lambda} \left( \ep{ \iu ( \phi x_1-{\Phi})}  c^{*}_{x+\hat e_2} c_{x} + c^{*}_{x+\hat e_1} c_{x} + \mathrm{h.c.} \right) -\mu \sum_{x \in \Lambda}  q_{x} + \sum_{x,y \in \Lambda}u(x-y)  q_{x}q_{y}.
\end{equation}
We have written $q_{x} = c_{x}\str c_{x}$ for the occupation operators and the parameters $t,\mu,u(\cdot)$ are, respectively, the hopping strength, the chemical potential and the interaction potential. We impose $L_1 \phi \in 2\pi\mathbb{Z}$ in order that the Hamiltonian is well defined on the torus.

\begin{figure}
  \begin{center}
    \includegraphics[width=0.35\textwidth]{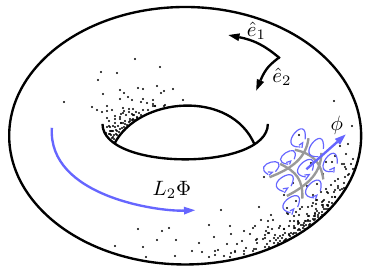}
  \end{center}
  \caption{The parameter $\phi$ is the magnetic flux piercing each unit cell on the torus, leading to a constant magnetic field. 
In contrast,  $L_2\Phi$ is the total flux threaded through a hole of the torus, it does not result in any magnetic field on the surface of the torus. }
\label{fig: harperhubbard}
\end{figure}

Let us comment on how this model fits our setup. The on-site operator $q_x$ is a concrete example of the general charge introduced in Section~\ref{sec:charge}. For $u=0$, the model is non-interacting and it is the Harper model in its second quantized version, also known as the Hofstadter model \cite{hofstadter1976energy}. By choosing $\mu$ to lie in one of the gaps of the corresponding one-particle Schr\"odinger operator, one obtains a gapped many-body Hamiltonian. In that case, fermionic perturbation theory \cite{giuliani2017universality,hastings2019stability,DeRoeck2018}  yields that the gap remains open for sufficiently weak $u$ and so Assumption \ref{ass: gap} is satisfied. This case corresponds to $p=1$. 
Although rigorous results are absent, it is believed that at strong interaction, the system exhibits topological order, hence $p>1$, and it is a fractional Quantum Hall insulator, see for example~\cite{PrangeGirvin}. The result below applies to that case as well, but we cannot establish the validity of Assumption \ref{ass: gap}. Below, we introduce two concrete unitaries that fit the general framework, namely translation and magnetic flux insertion.

\subsubsection{Translation and magnetic translation}\label{sec: magnetic translations}
The Hamiltonian (\ref{eq: hubbard}) is translation-invariant in the $x_2$-direction, but it is more interesting to consider translation in the $x_1$-direction
$$
\Theta^* c_x \Theta= c_{x+\hat e_1},\qquad   \Theta^* c^*_x \Theta= c^*_{x+\hat e_1},
$$ 
as well as the magnetic translation defined by 
$$
U^* c_x U= c_{x+\hat e_1}\ep{-\iu x_2\phi },\qquad   U^* c^*_x U= c^*_{x+\hat e_1} \ep{\iu x_2 \phi }.
$$ 
At $\phi\neq 0$, the ordinary translation is not a symmetry, but the magnetic translation is a symmetry provided that $L_2\phi\in 2\pi\bbZ$:
$$[H_{\Phi},\Theta] \neq 0,\qquad [H_{\Phi},U] = 0.$$
In that case, one can apply the index theorem with $U$ being magnetic translation and the conserved charge being the fermion number. Since, in the notation used before Corollary~\ref{cor:LSM}, we still have $T_-=Q_{[0]}$, we recover the result of the corollary. However, the conclusion~(\ref{eq: lsm physical}) fails for a non integer flux $\phi/(2\pi)$ because $L_2$ must be a multiple of $n$. Therefore the strongest result that can be obtained on the density $\rho$ is simply
$(pn)\rho \in \bbZ$, which also follows from an application of the theorem with $\Theta^n$. 
A simple check with free fermions confirms indeed that no sharper result is possible, at least not for $p=1$. In other words, in the case of magnetic systems, the density that satisfies~(\ref{eq: lsm physical}) is not the charge per unit cell, but the charge \emph{per magnetic unit cell} which is $n$ times larger. 

\subsubsection{Flux threading and Hall conductance}\label{sec: flux threading}
As illustrated in Figure~\ref{fig: harperhubbard}, the parameter $\Phi$ is the flux threaded through the torus, per unit length in $x_2$-direction. If the gap of $H_{\Phi}$ remains bounded away from $0$ as the parameter $\Phi$ goes from $\Phi$ to $\Phi'$, then the ground state projector $P_{\Phi}$ can be parallel transported to $ P_{\Phi'}$ by a locality-preserving unitary $F(\Phi,\Phi')$ (see Section~\ref{sec: general framework} for details) namely
\begin{equation}\label{eq: parallel}
F(\Phi,\Phi')\str P_{\Phi} F(\Phi,\Phi') =P_{\Phi'}.
\end{equation}
If the total threaded flux is an integer number of elementary flux quanta, i.e.\ $L_2(\Phi'-\Phi)\in 2\pi\bbZ$, then the effect of parallel transport on $P_{\Phi}$ is the same as that of a gauge transformation\footnote{It is often stressed that this is a `large' gauge transformation, referring to  the fact that it is not connected to the identity within the gauge group. However, on the lattice, it is not straightforward to make this distinction precise.} implemented by the unitary 
$$\caF_{\Delta\Phi}=\ep{-\iu \Delta\Phi \sum_{x\in\Lambda} {x_2} n_x}, \qquad \Delta\Phi=\Phi'-\Phi.$$ This follows because in that case
\begin{equation}\label{eq: phionephitwo}
\caF_{\Delta\Phi} H_{\Phi}\caF_{\Delta\Phi}\str =H_{\Phi' }.
\end{equation}
Combining (\ref{eq: parallel}) and (\ref{eq: phionephitwo}), we obtain indeed
$$
  [  P_{\Phi},  F(\Phi,\Phi')\caF_{\Delta\Phi}]=0.
$$
Therefore, the unitary $F(\Phi,\Phi')\caF_{\Delta\Phi}$ satisfies Assumption~\ref{ass: unitary} and it defines an associated index of $H_{\Phi}$. By the Laughlin argument, this index is the number of threaded flux quanta $L_2\Delta\Phi/(2\pi)$ times the quantum Hall conductance. This is discussed in~\cite{BBDF}, where we also provide an explicit proof relating this index (more precisely, an equivalent one) to the adiabatic curvature and \cite{bachmann2018note} details the relation of the adiabatic curvature to other expressions for Hall conductance in a many-body setting.

We conclude by noting that the convenient choice $\Delta\Phi\in 2\pi\bbZ$ yields $\caF_{\Delta\Phi}=1$, in which case $F(\Phi,\Phi')$ itself is a symmetry of $P_{\Phi}$. This will be exploited below. 

\subsubsection{The Avron-Dana-Zak relation}\label{sec: adz}

We are now equipped to obtain a relation between Hall conductance and charge density, taking advantage of the fact that flux threading can be intertwined with translations to yield a new symmetry, quite similarly to the case of magnetic translations. Indeed, the Hamiltonian~(\ref{eq: hubbard}) is covariant in the sense that $ H_{\Phi+\phi}  = \Theta\str H_{\Phi}\Theta $. Using this we get
$$
 P_{\Phi}=\Theta  P_{\Phi+\phi}   \Theta^*  =  \Theta F(\Phi+\phi,\Phi) P_{\Phi}  F(\Phi+\phi,\Phi)\str  \Theta^* 
$$  
so that $U : =\Theta F(\Phi+\phi,\Phi) $ is a symmetry satisfying Assumption~\ref{ass: unitary}. 
We shall later establish that
$$U^n= \Theta^n F(\Phi+n\phi,\Phi),$$
see Lemma~\ref{lma:Commuting TF}. By the remarks above and~(\ref{frac phi}), the two unitaries on the right-hand side are now symmetries in their own right. By additivity of the index, Proposition \ref{Prop:Additivity}, we conclude that
$$
\mathrm{Ind}_P(\Theta^n)+  \mathrm{Ind}_P( F(\Phi+n\phi,\Phi))  \in \frac{n}{p}\fatlattice,\quad P = P_\Phi.
$$
The first term on the left is $\langle Q_{[0,n-1]}\rangle_P=\sum_{x: 0\leq x_2<n} \langle n_x\rangle_P $, which is $n$ times the average density per slab, namely $nL_2\rho$, and the second term is $n\phi L_2$ times the Hall conductance $\sigma_L$, see again~\cite{BBDF}.
Since $L_2$ is arbitrary and if the convergence of $\rho_L,\sigma_L$ to their infinite-volume limits $\rho,\sigma$ satisfies $L_2(\rho-\rho_L),L_2(\sigma-\sigma_L) \to0 $, then the argument of Section~\ref{sec: flsm} implies that 
$$
\rho+ \phi\sigma \in \frac{1}{p}\bbZ,
$$
which is the fractional ADZ relation.

\subsubsection{A many-body Avron-Dana-Zak theorem}\label{sec: general framework}

In this concluding section, we repeat the discussion above under detailed assumptions and complete the proofs. The setting as in Section~\ref{sec: flsm}, with in particular the graph product
$
{\totalgraph} = {\mathbb{T}}\,\Box\, {\totalgraph}',
$
where ${\mathbb{T}}$ is the discrete circle with $L_1$ sites, and we again assume that $L_1>cL$. The unitary $\T$ is a translation along $\bbT$. The operator $Q$ continues to be the charge in a half (along $\bbT$) system.
We consider a family of local Hamiltonians $\{H_\Phi:\Phi\in\bbR\}$ of the form~(\ref{local Hamiltonian}),  satisfying 
\begin{assumption}\label{ass: Covariance}
\hspace{0.001mm}
\begin{enumerate}
\item  The parameters $R_H,m_H$ can be chosen uniform in $\Phi$ and the local terms $h_Z$ are  themselves $C^1$ functions of $\Phi$, such that $\Vert \partial_\Phi h_Z(\Phi) \Vert$ is bounded uniformly in $Z$ and $L$. 
\item Periodicity:
\begin{equation*}
H_{\Phi + 2\pi} = H_{\Phi}.
\end{equation*}
\item Covariance:
\begin{equation}\label{Assumption: covariance}
\T\str H_{\Phi} \T = H_{\Phi+\phi},
\end{equation}
 where 
\begin{equation*}
\phi = 2\pi \frac{m}{n},  \qquad m \in \bbZ, n \in \bbN, \quad \text{$m,n$ coprime}.
\end{equation*}
\item Assumption \ref{ass: gap} holds uniformly in $\Phi$, where $P=P_\Phi$ are the ground state projectors. This implies in particular that the rank $\fraction=\mathrm{rk}(P_\Phi)$ is independent of $\Phi$.
\end{enumerate}
\end{assumption}
\noindent While the assumption is obviously motivated by the Hamiltonian~(\ref{eq: hubbard}) and the corresponding physical phenomenology, $H_\Phi$ below is not restricted to that specific form; we only impose that Assumption~\ref{ass: Covariance} holds.

Items (i,iv) imply  that $\Phi\mapsto P_\Phi$ is itself differentiable and periodic,
\begin{equation}\label{P Gauge Invariance}
P_{\Phi + 2\pi} = P_{\Phi}
\end{equation}
for any $\Phi\in\bbR$. Items (ii,iii) lead to  
\begin{equation*}
[\T^n,P_\Phi] = 0,
\end{equation*}
and hence $\T^n$ satisfies Assumption~\ref{ass: unitary}.
Furthermore, it follows that there exists a locality-preserving, charge conserving unitary propagator $F(\Phi,\Phi')$, see~\cite{HastingsWen,Sven} and the paragraph above (\ref{cocycle}), such that
\begin{equation}\label{Parallel transport}
P_\Phi F(\Phi,\Phi') = F(\Phi,\Phi') P_{\Phi'}
\end{equation}
exactly, for any $\Phi,\Phi'\in\bbR$. In particular, we get for $P := P_0$ that
\begin{equation*}
[F,P] = 0,
\end{equation*}
where $F := F(n\phi,0)$. In other words, the unitary $F$ also satisfies Assumption~\ref{ass: unitary}.

\begin{thm}\label{thm:Covariant Index}
Let Assumption~\ref{ass: Covariance} hold. Then
\begin{equation*}
\frac{\fraction}{n}\left(\mathrm{Ind}_P(\T^n) + \mathrm{Ind}_P(F) \right) \in\bbZ_{\caO(L^{-\infty})}.
\end{equation*}
\end{thm}

We now turn to the proof of Theorem~\ref{thm:Covariant Index}. By the gap assumption, the self-adjoint family $A_\Phi$ of so-called quasi-adiabatic generators defined by
\begin{equation}\label{QAGen}
A_\Phi := \int W(t) \ep{\iu t H_\Phi}\dot H_\Phi \ep{-\iu t H_\Phi} \dd t,
\end{equation}
where $\dot{} = \partial_\Phi$ and $W$ was introduced after~(\ref{HastingsGenerator}), are such that
\begin{equation*}
\iu \dot P_\Phi = [A_\Phi,P_\Phi].
\end{equation*}
The corresponding unitary propagator $F(\Phi,\Phi')$ is the unique solution of
\begin{equation*}
\iu\dot F(\Phi,\Phi') = A_\Phi F(\Phi,\Phi'),\qquad F(\Phi',\Phi') = 1,
\end{equation*}
and it is an intertwiner between $P_{\Phi'}$ and $P_{\Phi}$, see~(\ref{Parallel transport}). It satisfies the cocycle relation
\begin{equation}\label{cocycle}
F(\Phi'',\Phi)F(\Phi,\Phi') = F(\Phi'',\Phi').
\end{equation}
The covariance assumption and~(\ref{QAGen}) imply the following relation between translations and quasi-adiabatic propagator.

\begin{lemma}\label{lma:covariance}
For any $k\in\bbZ$ and any $\Phi,\Phi'\in\bbR$,
\begin{equation*}
F(\Phi,\Phi') \T^k =  \T^k F(\Phi+k\phi,\Phi'+k\phi).
\end{equation*}
\end{lemma}
\begin{proof}
Assumption~(\ref{ass: Covariance}) implies that $\T\str \dot H_\Phi \T = \dot H_{\Phi + \phi}$, and with~(\ref{QAGen})
\begin{equation*}
\T\str A_\Phi \T = A_{\Phi+\phi }.
\end{equation*}
Hence, the operator $F_\T(\Phi) := \T\str F(\Phi,\Phi') \T$ is a solution of
\begin{equation*}
\iu\dot F_\T(\Phi) = A_{\Phi+\phi} F_\T (\Phi),\qquad F_\T(\Phi') = 1.
\end{equation*}
By uniqueness of the solution of the differential equation, we conclude that $F_\T(\Phi) = F(\Phi+\phi,\Phi'+\phi)$. In other words,
\begin{equation*}
F(\Phi,\Phi') \T =  \T F(\Phi+\phi,\Phi'+\phi ),
\end{equation*}
a $k$-fold application of which yields the claim.
\end{proof}
\begin{lemma}\label{lma:Commuting TF}
Let $U (\Phi) := \T F(\Phi+\phi,\Phi)$. Then, for any $\Phi\in\bbR$,
\begin{equation}\label{jointU}
[U(\Phi), P_{\Phi}] = 0. 
\end{equation}
Moreover, $[U(\Phi)^k, P_{\Phi}] = 0$, and
\begin{equation}\label{U^k}
U(\Phi)^k = \T^k F(\Phi + k\phi,\Phi)
\end{equation}
for any $k\in\bbZ$.
\end{lemma}
\begin{proof}

Assumption~(\ref{ass: Covariance}) implies $\T P_{\Phi +\phi } \T^* = P_\Phi$ and hence by (\ref{Parallel transport}) we have
\begin{equation*}
\T F(\Phi+\phi ,\Phi) P_\Phi = \T P_{\Phi + \phi }  F(\Phi+\phi ,\Phi) = P_\Phi \T F(\Phi+\phi ,\Phi).
\end{equation*}
This proves~(\ref{jointU}), and $U(\Phi)^k P_{\Phi} = P_\Phi U(\Phi)^k$ follows by $k$-fold application. We show that (\ref{U^k}) holds by induction. Assuming that the formula holds for $k-1$, we have
\begin{equation*}
U(\Phi)^k = U(\Phi) \T^{k-1} F(\Phi + (k-1)\phi ,\Phi).
\end{equation*}
By Lemma~\ref{lma:covariance}, we have
\begin{equation*}
\T^{k-1} F(\Phi + (k-1)\phi ,\Phi) = F(\Phi,\Phi- (k-1)\phi ) \T^{k-1}
\end{equation*}
and hence by the cocycle property~(\ref{cocycle}) of the propagator
\begin{equation*}
U(\Phi)^k = \T F(\Phi+\phi ,\Phi- (k-1)\phi ) \T^{k-1}.
\end{equation*}
Another application of Lemma~\ref{lma:covariance} yields the claim. 
\end{proof}
We note that by Lemma~\ref{lma:covariance},
\begin{equation*}
U(\Phi) = F(\Phi,\Phi-\phi )\T,\qquad 
U(\Phi)^k = F(\Phi ,\Phi - k\phi )\T^k.
\end{equation*}

\begin{proof}[Proof of Theorem~\ref{thm:Covariant Index}]
Since $n \phi \in\bbZ$, we have that
\begin{equation*}
F(\Phi +  n\phi ,\Phi ) P_\Phi  = P_\Phi F(\Phi + n\phi ,\Phi )
\end{equation*}
by~(\ref{P Gauge Invariance}). Setting $k=n$ in Lemma~\ref{lma:Commuting TF}, we further conclude that $P_\Phi$ is also invariant under $\T^n$. It follows from Theorem~\ref{thm:Index} that both indices $\mathrm{Ind}_{P_\Phi}(\T^n),\: \mathrm{Ind}_{P_\Phi}(F)$ are well defined with denominator $\fraction$. Finally, the same applies to $U(\Phi)$ and $U(\Phi)^n$ by Lemma~\ref{lma:Commuting TF}. Altogether,
\begin{align*}
\mathrm{Ind}_{P_\Phi}(\T^n) + \mathrm{Ind}_{P_\Phi}(F ) 
& \eqL \mathrm{Ind}_{P_\Phi}(\T^n F) \\
& \eqL \mathrm{Ind}_{P_\Phi}(U(\Phi)^n) \\
& \eqL n \mathrm{Ind}_{P_\Phi}(U(\Phi)),
\end{align*}
where the first equality is by Proposition~\ref{Prop:Additivity}, the second one is by Lemma~(\ref{lma:Commuting TF}) and the last one is additivity again. Since $\mathrm{Ind}_{P_\Phi}(U(\Phi))  \in \fraction^{-1}\bbZ_{\caO(L^{-\infty})}$, the theorem follows.
\end{proof}

\section{Superselection sectors}\label{sec:sectors}

In many physically relevant situations, the topological order assumption is violated, and ground states can de distinguished by a local order parameter. The prime example thereof is dimerization, and more generally the breaking of a discrete symmetry. The projection $P$ decomposes as
\begin{equation}\label{sectors}
P=\oplus_{m=1}^M P_m,
\end{equation}
where $P_m = P_m\str = P_m^2$ and $\ran(P_m)$ are extremal in a sense to be made clear below. We refer to $\ran(P_m)$ as the \emph{superselection sectors} of the system and let
\begin{equation*}
\fraction_m := \rk(P_m).
\end{equation*}
As usual, $M,\fraction_m$ are finite and fixed, independent of $L$ for $L$ large enough. We denote $\langle \cdot \rangle_m:=\fraction_m^{-1}\Tr (P_m \cdot)$.

 \begin{assumption} \label{ass: superselection}
 The decomposition~(\ref{sectors}) is such that:
\begin{enumerate}
\item  For any $O\in\caA_Z$ of norm $1$ and with $\diam(Z) \leq C$,
\begin{equation}\label{Orthogonal sectors}
P_{m'}OP_m \eqL 0
\end{equation}
whenever $m\neq m'$.
\item  There are self-adjoint observables $A_1\in\caA_X, A_2\in\caA_Y$ with $\mathrm{diam}(X)\leq C$, $\mathrm{diam}(Y) \leq C$ and $\dist(X,Y) >cL$, such that
$$
\langle A_1 \rangle_{m} = \langle A_2 \rangle_{m},
$$
for all $m$, while 
\begin{equation*}
\vert \langle A_1 \rangle_{m} - \langle A_1 \rangle_{m'}\vert \geq c>0,\qquad m\neq m'.
\end{equation*}
\item Each $P_m$ satisfies the topological order Assumption \ref{ass: toporder}.
\end{enumerate} 
\end{assumption}

\begin{rem}
Assumption~(ii) postulates the existence of observables that detect the local order. In the case of a translation-invariant system, $A_2$ is a translate of $A_1$. In other situations, $A_2$ might be a linear combination of more natural physical observables that achieves the equality of expectation values between $A_2$ and $A_1$. Points~(i) and~(iii) of the assumptions imply that $P_m$ are extremal in the sense that a finer decomposition of $P$ also satisfying~(i) and~(iii) cannot exist.
\end{rem}

These assumptions imply the following specific form of the unitary $U$. Let $S_M$ be the symmetric group on $\{1,\ldots,M\}$.

 \begin{lemma}\label{lma: bijection u}
Let $\caU_{m,l}:=P_m U P_l$, and $\caU := PUP$. There is a permutation $\pi_U\in S_M$ such that
$$
\caU \eqL \oplus_{m=1}^M \caU_{\pi_U(m),m}.
$$
In particular, 
\begin{align*}
(i) \quad &\caU_{l,m} \eqL 0\text{ unless }l=\pi_U(m),\\
(ii) \quad &\rk(P_{\pi_U(m)})=\rk(P_{m}), \\ 
(iii) \quad &U P_m U\str \eqL P_{\pi_U(m)}.
\end{align*}
\end{lemma}
\begin{proof} 
Since $[U,P] \eqL 0$ and $P_m$ are subprojections of $P$,
\begin{equation*}
\Tr(UP_mU\str O) \eqL \Tr(UP_mU\str  POP).
\end{equation*}
Since $P = \oplus_{l=1}^M P_l$, Assumption~\ref{ass: superselection}(i) now implies that for any $O$ satisfying the assumption
\begin{equation*}
\Tr(UP_mU\str O) \eqL \sum_{l=1}^M \Tr(UP_mU\str  P_lOP_l).
\end{equation*}
By~(iii) of the assumption, we conclude that
\begin{equation}\label{U proba}
  \langle U\str O U\rangle_m  \eqL \sum_{l = 1}^M \rho_m(l) \langle O\rangle_l.
\end{equation}
where, for any $1\leq m,l\leq M$,
\begin{equation*}
\rho_m(l) := \langle U\str P_l U\rangle_m = \langle (\caU_{l,m})\str \caU_{l,m}\rangle_m.
\end{equation*}
For any fixed $m$, $\rho_m$ is a probability distribution on $\{1,\ldots,M\}$, and for $O = A_j$ as in Assumption~\ref{ass: superselection}(ii),  we interpret (\ref{U proba}) as the expectation value in $\rho_m$ of a random variable $a_j:\{1,\ldots,M\}\to\bbR$ given by $a_j(l) = \langle A_j\rangle_l$. Therefore,
\begin{equation*}
\langle U\str A_j U \rangle_m = \bbE_m(a_j).
\end{equation*}
where $\bbE_m$ is the expectation value associated to $\rho_m$. In terms of the random variables $a_1,a_2$, Assumption~~\ref{ass: superselection}(ii) is rephrased as the statement that $a_1 = a_2$, and that $a_1$ is injective. Clustering, see Lemma~\ref{lma:subclustering} below, implies that
\begin{equation*}
\bbE_m(a_1a_2) \eqL \bbE_m(a_1)\bbE_m(a_2),
\end{equation*}
but since, $a_1=a_2$, we simply get
\begin{equation*}
\bbE_m(a^2_1) \eqL \bbE_m(a_1)^2.
\end{equation*}
We conclude that $a_1$ is constant on the intersection of its support with the support of $\rho_m$.  Since $a_1$ is injective, the support of $\rho_m$ is a singleton. 
This means that there is a unique $\pi(m)$ such that $\caU_{\pi(m),m}$ is non-vanishing. The invertibility of $U$ on $\ran(P)$ and $(\caU\str)_{l,m} = (\caU_{m,l})\str$ imply that $\pi$ is a bijection, proving~(i). With this,
\begin{equation*}
U P_m U\str 
=P_{\pi_U(m)} U P_m U\str P_{\pi_U(m)}
\end{equation*}
so that $U P_m U\str$ is a subprojection of $P_{\pi_U(m)}$. By the invertibility of $U$, the two must have the same rank, proving~(ii) and hence they are equal, proving~(iii).
\end{proof}
 For any $1\leq m\leq M$, let $(\pi_U\cdot m)$ denote the the cycle of the permutation $\pi_U$ containing $m$, and let $\ell_U(m)$ be its length. We then have the following generalizations of Theorem~\ref{thm:Index}.
\begin{prop} \label{prop:LTQO in P_m}
Let Assumptions \ref{ass: gap}, \ref{ass: unitary} and \ref{ass: superselection} hold. Then \\
(i) For any normalized $\Psi_m \in \ran(P_m)$,
\begin{equation*}
\fraction_m \Big\langle \Psi_m, T\big(U^{\ell_U(m)}\big)_- \Psi_m \Big\rangle  \in \fatlattice.
\end{equation*}
(ii) For any normalized $\Psi \in \oplus_{m'\in(\pi_U\cdot m)}\ran(P_{m'})$,
\begin{equation*}
\ell_U(m)\fraction_m \Big\langle \Psi, T\big(U\big)_- \Psi \Big\rangle  \in \fatlattice.
\end{equation*}
\end{prop}

We point out that the proof of Theorem~\ref{thm:Index} presented in Section~\ref{Sec: Main} does not require $P$ to satisfy Assumption~\ref{ass: gap} per se. It uses only two consequences thereof, namely (i) the clustering property~(\ref{Assumption: Clustering}) and (ii) and the invariance of $\ran(P)$ under $\overline Q$. The two lemmas below show that Assumption~\ref{ass: superselection} implies both properties for the subprojections $P_m$. 
\begin{lemma}\label{lma:subclustering}
Let $\Psi_m\in\mathrm{Ran}(P_m)$ be normalized and let $A\in\caA_X,B\in\caA_Y$ be of norm $1$ and with $d(X,Y)>cL$. If one of $A,B$ is an observable as in Assumption~\ref{ass: superselection}(i), then 
\begin{equation*}
\langle \Psi_m,A(1-P_m)B\Psi_m\rangle
\eqL 0.
\end{equation*}
\end{lemma}
\begin{proof}
$P_m$ being a subprojection of $P$, clustering~(\ref{Assumption: Clustering}) implies that $\langle \Psi_m,A(1-P)B\Psi_m\rangle \eqL 0$. But then $P$ can be replaced by $P_m$ by~(\ref{sectors},\ref{Orthogonal sectors}). 
\end{proof}
\begin{lemma}
For any $m$,  $[\overline{Q},P_m] \eqL 0$.
\end{lemma}
\begin{proof}
The equality $[\overline Q,P]\eqL 0$ holds by construction. But $P_m = PP_mP$ implies
\begin{equation*}
[\overline Q,P_m] = P [\overline Q,P_m] P \eqL P_m[\overline Q,P_m] P_m \eqL 0,
\end{equation*}
where the second equality is by Assumption~\ref{ass: superselection}(i) and the fact that $\overline Q$ is a sum of local terms.
\end{proof}

\begin{proof}[Proof of Proposition~\ref{prop:LTQO in P_m}]
By the above lemmas $P_m$ satisfies the clustering property and it is invariant under $\overline{Q}$ . Furthermore, the unitary $U^{\ell_U(m)}$ keeps $P_m$ invariant by definition of $\ell_U(m)$, namely
\begin{equation*}
[U^{\ell_U(m)},P_m] \eqL 0.
\end{equation*}
The proof and hence the result of Theorem~\ref{thm:Index} carries step by step through with $P$ replaced by $P_m$, and $U$ replaced by $U^{\ell_U(m)}$. This establishes~\emph{(i)}.

Since $(\pi_U\cdot m)$ is a cycle, $\oplus_{m'\in(\pi_U\cdot m)}\ran(P)$ is invariant under $U$. By Lemma~\ref{lma: bijection u}(ii), all factors have rank $\fraction_m$, so that the dimension of $\oplus_{m'\in(\pi_U\cdot m)}\ran(P)$ is $\ell_U(m) \fraction_m $ and the claim follows as above.
\end{proof}

\section*{Acknowledgements}

\noindent The authors would like to thank an anonymous referee for his comments on the clustering theorem which lead us to provide Proposition~\ref{prop:clustering}. The work of S.B. was supported by NSERC of Canada. M.F. was supported in part by the NSF under grant DMS-1907435. W.D.R.\ thanks the Flemish Research Fund (FWO) for support via grants G076216N and 
G098919N. A.B. was supported by VILLIUM FONDEN through the QMATH Center of Excellence (Grant No. 10059).

\section*{Data Availability}

Data sharing is not applicable to this article as no new data were created or analysed in this study.

\appendix

\section{Clustering in finite volume} \label{A:Clustering}

We prove Proposition~\ref{prop:clustering} under Assumption~\ref{ass: gap}. First of all, $\gamma>2\Delta$ implies $-\gamma+\Delta<-\Delta$. We consider a smooth, real-valued function $g$ that satisfies
\begin{equation}\label{g}
g(\omega)=\begin{cases}    0 &   \omega \geq -\Delta,  \\
1 &   \omega \leq -\gamma+ \Delta. \end{cases}
\end{equation}
For concreteness we take 
$$
g=\phi \star \theta_{(-\gamma/2)}
$$
where, for any $a\in\bbR$, $1-\theta_a$ is the shifted Heaviside function (with discontinuity at~$a$), $-\gamma/2$ lies in the interval $(-\gamma+\Delta,-\Delta)$, the product $\star$ denotes the convolution and $\phi\in C_c^\infty((-\delta,\delta))$ is a non-negative function such that $\int \phi=1$.

Let
\begin{equation*}
\caQ(O) =    g(\mathrm{\adjoint}_{H})(O)
\end{equation*}
where $\mathrm{\adjoint}_{H}(O) = [H,O]$. This is well-defined by functional calculus for self-adjoint matrices, given that the operator $O\mapsto [H,O]$ is self-adjoint with respect to the Hilbert-Schmidt inner product. We claim, and shall prove below, that the Fourier transform of $g$ is a tempered distribution whose singular support is $\{0\}$ and whose regular part vanishes at infinity as $\widehat g(t) = \caO(\vert t\vert^{-\infty})$, and so
\begin{equation}\label{Q of O}
\caQ(O) = \frac{1}{\sqrt{2\pi}}\int_{-\infty}^\infty  \widehat g (t)    \ep{\iu t H} O \ep{-\iu t H} dt
\end{equation}
is well defined. This expression further implies by the Lieb-Robinson bound that $\caQ(O)$ is almost supported in the support of $O$, see again below.

With this,
\begin{equation*}
\caQ(O) P = 0\quad\text{and}\quad P\caQ(O) = PO(1-P).
\end{equation*}
This follows from~(\ref{Q of O}), namely
\begin{equation*}
\caQ(O) = \sum_{j,k}g(E_j - E_k) P_j O P_k,
\end{equation*}
where $P_l$ is the eigenprojection corresponding to the eigenvalue $E_l$, and~(\ref{g}). Hence, for any $A\in\caA_X,B\in\caA_Y$,
\begin{align*}
\langle A B \rangle  &= \langle A P B \rangle+ \langle A (1-P) B \rangle  \\
 &= \langle A P B \rangle+ \langle \caQ(A)  B \rangle  \\
  &= \langle A P B \rangle+ \langle [\caQ(A),  B] \rangle.
\end{align*}
It remains to note that $\caQ(A)\in\caA_X$ implies $[\caQ(A),  B] = \caO(d(X,Y)^{-\infty})$. This yields the claimed clustering result.

We finally turn to the technical questions left open above. The Fourier transform of $g$ is the tempered distribution $\widehat g$ given by
$$
\widehat g  = \widehat \phi \cdot \widehat \theta_{(-\gamma/2)}, \qquad  \widehat \theta_{(-\gamma/2)} = \ep{-\iu t {\gamma/2}}\widehat\theta_0 .
$$
The function $\phi$ being smooth, we have that $\widehat \phi(t) = \caO(\vert t\vert^{-\infty})$.  Moreover, the Plemelj-Sochotcki formula  implies that
\begin{equation*}
\widehat\theta_0 = \sqrt\frac{\pi}{2}\, \delta + \frac{\iu}{\sqrt{2\pi}}\caP\left(\frac{1}{t} \right),
\end{equation*}
where $\caP$ denotes the principle value, and we recall that
\begin{equation}\label{Schwartz}
\vert  \widehat \theta_0 [\psi]  \vert   \leq \sqrt{\pi/2}\Vert \psi\Vert_{\infty} + \sqrt{2/\pi}\Vert \psi'\Vert_{\infty}
\end{equation}
for any Schwartz function $\psi$.

The smooth operator-valued function $O(t) = \ep{\iu t H} O \ep{-\iu t H}$ is bounded since $\Vert O(t)\Vert = \Vert O \Vert$, and its derivative $O'(t) = \ep{\iu t H} \iu[H,O] \ep{-\iu t H}$ is similarly bounded: $\Vert O'(t) \Vert \leq C \Vert O\Vert \vert \mathrm{supp}(O)\vert$. It follows in particular that $\ep{-\iu t {\gamma/2}}\widehat\phi(t)O(t)$ is a Schwartz operator-valued function. The expression on the right of~(\ref{Q of O}) is well-defined for any local observable~$O$, with
\begin{equation*}
\int_{-\infty}^\infty  \widehat g(t)   \ep{\iu t H} O \ep{-\iu t H} dt = \widehat \theta_0 [\ep{-\iu t {\gamma/2}}\widehat\phi(t)O(t)]
\end{equation*}
and the bound~(\ref{Schwartz}) holds in operator norm.

We can finally address the locality of $\caQ(O)$. We assume that $O$ is a local operator, the general case of an almost local one follows by approximation. The $\delta$ contribution in~$\widehat \theta_0$ is supported on the support of $O$. As for the $\caP\left(\frac{1}{t} \right)$ contribution, the integral is split as $\vert t\vert\geq T$ and $\vert t \vert< T$. The fast decay of $\widehat\phi$ yields a bound $\Vert O \Vert\caO(T^{-\infty})$ for the long time part. For the second, short time part, we write $O(t) = \Pi_R(O(t)) + (O(t) - \Pi_R(O(t)))$, where $\Pi_R$ is the projection\footnote{It is the normalized partial trace in the case of a quantum spin system, while a similar projection can be constructed for lattice fermions, see~\cite{NSY_Fermions}.} onto an $R$-fattening of the support of $O$. Since $\Pi_R(O(t))$ is smooth, the first contributions is well-defined, bounded uniformly in $T$, and local by construction. For the second one, the Lieb-Robinson bound yields
\begin{equation*}
\frac{1}{t}\Vert O(t) - \Pi_R(O(t))\Vert \leq C\Vert O\Vert\vert\mathrm{supp}(O)\vert\ep{-\xi R}\,\frac{\ep{c v \vert t\vert}-1}{t},
\end{equation*}
which is integrable at $0$ with $\int_{[-T,T]}t^{-1}\left(\ep{c v \vert t\vert}-1\right)dt\leq 2(\ep{cv T}-1)$.
It remains to pick $T = \frac{\xi}{2cv} R$ to conclude that both terms are $\caO(R^{-\infty})$. Altogether, we conclude that $\caQ(O)$ is almost localized on the support of $O$. 

\section{Notions of locality} \label{app: locality}

We consider first the case of spin systems. Then $\caH_\totalgraph \cong \otimes_{x\in\totalgraph}\bbC^n$  and $\caA \cong \otimes_{x\in\totalgraph} \bbC^{n^2}$, and so there is a natural tensor product representation  $\caA\cong  \otimes_{x\in S} \bbC^{n^2} \otimes \otimes_{x\in S^c} \bbC^{n^2}$, for any $S \subset \totalgraph$. 
The support of $O \in \caA$ is then the smallest set $S$ such that $O \cong O_S \otimes 1_{S^c}$ in this natural representation. We refer to Section~2.6 of~\cite{BRI}, see also Section~5.2.2.1 and Section~6.2.1 of~\cite{BRII} for further detailed discussions.

For fermionic systems a bit more setup is needed. The total algebra of operators is generated by $c_{x,\alpha},c\str_{x,\alpha}$ and the identity,  with $c_{x,\alpha},c\str_{x,\alpha}$ the annihilation/creation operator of a fermion at site $x\in\totalgraph$ and with  label $\alpha \in \{1,\ldots, f\}$. The label can correspond to spin states or something else. These operators satisfy the CAR
$$
\{ c_{x,\alpha},c\str_{x',\alpha'}\}=\delta_{x,x'} \delta_{\alpha,\alpha'}, \qquad   \{ c_{x,\alpha},c_{x',\alpha'}\}=\{ c\str_{x,\alpha},c\str_{x',\alpha'}\}=0.
$$
The Hilbert space $\caH_\totalgraph$ is generated by acting with these operators on the vacuum state $\Omega$, which is the common eigenvector of all annihilation operators.
The total algebra of observables is graded by fermion parity: monomials in the $c_{x,\alpha},c_{x,\alpha}\str$ of even/odd degree are called even/odd and the even monomials generate an algebra, which we call a $\caA$.
The support of $O\in \caA$ is then the smallest set $S$ such that $O$ is in the linear span of even monomials with $x$ restricted to $S$.

\section{The large $L$ asymptotic analysis}
\label{app:L}
In the main text we have systematically omitted the index $L$ but at the same time used symbols like $\vsmall$ and $\mathcal{A}_Z$ that make sense only for a sequence of operators labeled by $L$. We believe that the chance of misunderstanding coming from this convention is small, but we now provide more details for the sake of completeness.

Let $A_L$ be a sequence of operators defined on  a sequence of operator algebras $\mathcal{A}=\mathcal{A}_{L}$ associated with a sequence of graphs ${\totalgraph}_L$ with diameter $L$. Neither the graphs nor the operators $A_L$ are a priori related in any way. We say that $A_L = \vsmall$ if for any $k\in\bbN$, there exists a constant $C_k$ such that $\| A_L \| \leq C_k L^{-k}$ for all $L$. Furthermore let $Z_L \subset {\totalgraph}_L$ be a sequence of sets. The sequence $A_L$ belongs to the set $\mathcal{A}_{Z_L}$ if for all $r$ there exists a sequence of operators $A_{L,r}$ supported in $(Z_{L})_{(r)}$ such that for any $k\in\bbN$, there exists a constant $C_k$ with
$$
\|A_L - A_{L,r} \| \leq  C_k \|A_L\| |Z_L| r^{-k}
$$ 
for all $L$.

In the setting of Section~\ref{sec: spatial}, we postulate the existence of a sequence of sets $\Gamma_L \subset {\totalgraph}_L$ with boundaries $\partial_{-,L}$ and $\partial_{+,L}$ whose distance satisfies (\ref{Spatial structure}). Assumption~\ref{ass: unitary}, in particular~(\ref{eq: splitting}), then says that $T_{\pm,L} \in \mathcal{A}_{\partial_{\pm,L}}$. Assumption~\ref{ass: gap} postulates the existence of a sequence of ground state projections $P_L \in \mathcal{A}_{L}$ of $L$-independent rank $\fraction$. Projections $P_L$ are further constrained by Assumption~\ref{ass: toporder}: for any sequence of operators $A_L \in \mathcal{A}_{Z_L}$ such that $\mathrm{diam}(Z_L) < C$ and $\|A_L\|=1$, we have $P_LA_LP_L = p^{-1}\Tr(P_L A_L) P_L + O(L^{-\infty})$. Theorem~\ref{thm:Index} states that there exists a sequence of integers $n_L$ such that for any sequence $\Psi_L\in\ran(P_L)$, 
$$
\left \vert \, p\, \langle \Psi_L ,T_{-,L}\Psi_L\rangle - n_L \right\vert = \vsmall. 
$$

As a final note, we point out that for an $L \times L$ torus, $\partial_{\pm, L}$ have $2L$ sites. If $U$ is a translation by one site in the $x_1$-direction then $T_{\pm, L}$ is equal to the charge in the boundary, see~(\ref{translation}), and $n_L$ is the total charge in this circle of length~$L$. This shows that in a physically interesting setting, all these sets, operators, algebras may indeed have a non-trivial dependence on $L$.

%

\end{document}